\documentclass[letterpaper, 10 pt, conference]{ieeeconf}  
\IEEEoverridecommandlockouts

\usepackage[pdftex]{graphicx}
\usepackage{amsmath} 
\usepackage{amssymb}  

\title{\LARGE \bf
Probabilistic Modeling Using Tree Linear Cascades
}

\author{Nicholas C. Landolfi$^{1}$ and Sanjay Lall$^{2}$
\thanks{$^{1}$Nicholas C. Landolfi is with the Department of Computer Science, Stanford University, Stanford, CA 94305
        {\tt\small lando@stanford.edu}}%
\thanks{$^{2}$Sanjay Lall is with the Department of Electrical Engineering, Stanford University, Stanford, CA 94305
        {\tt\small lall@stanford.edu}}%
}

\usepackage{caption} 
\usepackage{subcaption} 
\usepackage{nicefrac}
\newcommand{\citet}[1]{\cite{#1}}

\usepackage{cite}
\usepackage{xcolor}         
\newcommand{\term}[1]{{\it #1}}

\newcommand{\set}[1]{\left\{#1\right\}}
\newcommand{\Set}[2]{\left\{#1 \mid #2 \right\}}
\newcommand{\upto}[1]{\set{1, 2,\dots, #1}}
\newcommand{\pa}[1]{\text{pa}_{#1}}

\DeclareMathAlphabet{\mathbfsf}{\encodingdefault}{\sfdefault}{bx}{n}

\newcommand{\parens}[1]{\left(#1\right)}

\newcommand{\opt}[1]{#1^{\star}}

\newcommand{\inv}[1]{#1^{-1}}

\newcommand{\sparse}{\mathop{\mathrm{sparse}}}

\newcommand{\anc}{\mathop{\mathrm{anc}}}
\newcommand{\canc}{\mathop{\mathrm{canc}}}
\newcommand{\graphpath}{\mathop{\mathrm{path}}}

\newcommand{\cov}{\text{cov}}
\newcommand{\union}{\cup}

\newcommand{\appendixproof}[1]{See Appendix.}

\newcommand{\R}{\mathbb{R}}
\newcommand{\E}{\mathbb{E}}

\newcommand{\tp}{\top}

\newcommand{\abs}[1]{\left\lvert{#1}\right\rvert}
\newcommand{\norm}[1]{\left\lVert{#1}\right\rVert}

\newcommand{\diag}{\text{diag}}

\newtheorem{thm}{Theorem}
\newtheorem{theorem}[thm]{Theorem}

\newtheorem{lemma}{Lemma}

\newtheorem{problem}{Problem}

\newtheorem{corollary}{Corollary}

\usepackage[utf8]{inputenc} 
\usepackage[T1]{fontenc}    
\usepackage{hyperref}       
\usepackage{url}            
\usepackage{booktabs}       
\usepackage{amsfonts}       
\usepackage{nicefrac}       
\usepackage{microtype}      
\usepackage{xcolor}         

\begin{document}

\maketitle
\thispagestyle{empty}
\pagestyle{empty}

\begin{abstract}
We introduce tree linear cascades, a class of linear structural equation models for which the error variables are uncorrelated but need not be Gaussian nor independent.
We show that, in spite of this weak assumption, the tree structure of this class of models is identifiable.
In a similar vein, we introduce a constrained regression problem for fitting a tree-structured linear structural equation model and solve the problem analytically.
We connect these results to the classical Chow-Liu approach for Gaussian graphical models.
We conclude by giving an empirical-risk form of the regression and illustrating the computationally attractive implications of our theoretical results on a basic example involving stock prices.
\end{abstract}


\section{Introduction}
\label{section:introduction}

Throughout engineering and the sciences, one is often interested in modeling functional or causal relationships among high-dimensional multivariate data.
Examples range from molecular pathway modeling in genomics \cite{statnikov2012new} and fMRI brain imaging in neuroscience \cite{ramsey2010six}, to system performance monitoring in engineering \cite{landolfi2021cloud}.
In the context of control and decision problems, such aspects arise, for example, when modeling networks of dynamical systems \cite{materassi2010topological,saunderson2011tree, dimovska2017granger}.

Structural equation models and functional causal models are popular approaches \cite{pearl2009causality}.
These models associate each variable with a node in a directed acyclic graph and model the variable as a function of its parents (if any) in the graph and some random noise.
Both in theory and practice, one is interested in (a) identifying the graphical structure from observational data and (b) estimating the functional relations.

Not much can be said without assumptions on the class of models.
Consequently, functional assumptions (e.g., linearity, additive partial linearity) and  distributional assumptions (e.g., Gaussian, independent) are common.
Even so, the identifiability of the models is nontrivial \cite{rothenhausler2018causal}.
It is natural to look for subclasses of these models with favorable identifiability and estimation properties.

\paragraph{Tree Linear Cascades}
In this paper, we introduce a class of these models with linear functional relations and a directed tree graphical structure.
We prove that the tree structure is identifiable under the weak assumption that the error variables in the model are uncorrelated.

Given a tree $T$ and root vertex $r$, we call a random vector $x$ a \term{tree linear cascade} on an uncorrelated random vector $e$ with respect to a sparse matrix $A$ if $x$ satisfies
	$x = Ax + e$.
The sparsity pattern of $A$ matches that of the directed adjacency matrix of the rooted tree $(T, r)$.
For precise definitions, see Section~\ref{section:cascades}.
As usual with structural equation models, we interpret a tree linear cascade by associating the components of $x$ with the vertices of $T$.
Then $x_i$ is a linear function of its parent (if any) plus some noise $e_i$.

Parallel to finding favorable model classes, it is natural to look for favorable fitting methods.
A reasonable approach is to pose a problem that simultaneously finds a graphical structure and functional relations to give good fit.
Such formulations lead to computational difficulties, however, because the number of directed graphs is exponential in the number of variables and the estimation may be nontrivial.

\paragraph{Cascade Regression}
In this paper, we introduce and solve a regression problem to fit a tree-structured linear  structural equation model.
Although there are exponentially many directed trees, we show that the problem reduces to a tractable maximum spanning tree   problem.

Given a random vector $x$, \term{cascade regression} finds a rooted tree $(T, r)$ and set of coefficients $A$ matching the sparsity pattern of the directed adjacency matrix of $(T, r)$ to minimize $\E\norm{Ax - x}$.
For a precise definition, see Problem~\ref{problem:cascades:regression} of Section~\ref{section:regression:identification}.
In theory, this problem correctly recovers a tree linear cascade.
In practice, it has a natural instantiation as empirical risk minimization and the theoretical results lead to a computationally attractive practical technique.

Our solution to this problem is reminiscent of the classical results of Chow and Liu \cite{chow1968approximating}.
In particular, for the Gaussian density case, in which the solution also involves a maximum spanning tree problem involving squared correlations \cite{tan2010learninggaussian}.

\paragraph{Gaussian Chow-Liu}
Indeed, we show (easily, with our earlier results) that the trees found by both approaches coincide.
The Chow-Liu approach to approximating a density $g: \R^d \to \R$ by one which factors according to a tree finds a tree $T$ and a density $f$ factoring according to $T$ to minimize the Kullback-Leibler divergence $d_{kl}(g, f)$.
The famous solution to this problem produces the maximum spanning tree of a graph weighted by pairwise mutual informations of the components of $g$.
The connection to cascade regression suggests interpreting the regression as an approximation technique for modeling a high-dimensional distribution with a simpler and sparser low-dimensional one.

These results relate to a nice line of work in the control literature which explores reconstructing the topology of a network of interconnected dynamical systems from observational data \cite{materassi2013reconstruction,dimovska2017granger,materassi2019signal,sepehr2021noninvasive}.
That formulation leads to similar identifiability and estimation results for a set of linear dynamical systems with a tree connection structure \cite{materassi2010topological}.

\paragraph{Contributions}
The principal novel theoretical contributions of this paper are the statement and proof of the maximum spanning tree property for tree linear cascades (Theorem~\ref{theorem:probabilistic:mst}) and the posing and solution of cascade regression (Theorem~\ref{theorem:regression:optimalT}).
We also give interesting, but easier, corollaries connecting tree linear cascades to cascade regression (Corollary~\ref{corollary:regression:correctness}) and cascade regression to Gaussian Chow-Liu approximation (Corollary~\ref{corollary:chow:correspondence}).
To our knowledge, none of these results appears in prior literature.

\paragraph{Outline}
In Section~\ref{section:notation} we review notation and preliminaries.
In Section~\ref{section:related} we discuss related work.
In Section~\ref{section:cascades} we introduce tree linear cascades and prove their maximum spanning tree property.
In Section~\ref{section:regression} we pose and solve the simultaneous cascade regression problem.
In Section~\ref{section:chow} we review the classical results of Chow and Liu for the density case and give a novel connection to simultaneous cascade regression.
We give the empirical form of cascade regression and a basic illustration on stock prices in Section~\ref{section:empirical}.
We conclude in Section~\ref{section:conclusion}.

\subsection{Notation and preliminaries}
\label{section:notation}

As usual, a \term{density} on $\R^d$ is a function $f: \R^d \to \R$
satisfying $f \geq 0$ and $\int f = 1$.
For $j \neq i$ and $i,j \in \upto{d}$,
the
marginal densities $f_i: \R \to \R$ and $f_{ij}:\R^2 \to \R$ are defined by
\[
f_i(\xi) = \int_{x_i = \xi} f(x)dx,
\quad
f_{ij}(\xi, \gamma) = \int_{(x_i, x_j) = (\xi, \gamma)}f(x) dx.
\]
 The conditional density $f_{i \mid j}: \R^2 \to \R$ is defined to
satisfy $f_{ij} = f_{i \mid j} f_{j}$ for $j \neq i$ and $i, j \in
\upto{d}$.  As usual, a \term{Gaussian} density $f$ on $\R^d$ has mean
$\mu \in \R^d$ and positive definite covariance $\Sigma \in \R^{d
  \times d}$.
As usual, the \term{Kullback-Liebler divergence} of a density $f$ \term{relative} to a density $g$ is $d_{kl}(g, f) = h(g, f) - h(g)$ where $h(g) = -\int_{g(x) > 0} g(x) \log g(x) dx$ is the \term{differential entropy} and $h(g, f) =- \int_{f(x) > 0} g(x) \log f(x) dx$ is the \term{differential cross entropy} of $f$ relative to $g$.
For Gaussian $f$, $d_{kl}(f_{ij}, f_{i}f_{j}) = -(\nicefrac{1}{2})\log(1 - \Sigma_{ij}^2/\parens{\Sigma_{ii}\Sigma_{jj}})$.

As usual, we fix a \term{probability space} $(\Omega, \mathcal{A}, \mathbb{P})$.
A \term{random variable} is a measurable function $x: \Omega \to \R^d$.
The \term{expectation} of $x$ is $\E(x) = \int xd\mathbb{P}$.
The \term{covariance} of $x$ is $\cov(x) = \E\parens{x - \E x}\parens{x - \E x}^\tp$.
The \term{correlation} between components $i$ and $j$ of $x$ is $\E(x_ix_j)/\sqrt{\E(x_i^2)\E(x_j^2)}$.

As usual, a \term{tree} is a connected acyclic undirected graph.
The key property is that there is a unique path between any two vertices.
A \term{rooted tree} is a tree and a distinguished vertex which we call the \term{root}.
The first vertex $j$ on the path from $i$ to the root is the \term{parent} of $i$ and $i$ is the \term{child} of $j$.
Since each nonroot vertex has one parent, we write $\pa{i} = j$ to mean that the parent of the nonroot vertex $i$ is vertex $j$.
A vertex $j$ is an \term{ancestor} of $i$ if there is a directed path from $j$ to $i$.
We denote the set of ancestors of a vertex $i$ by $\anc(i)$.

A \term{density} $f$ \term{factors according to a tree $T$ rooted at vertex $i$} on $\upto{d}$ if $f = f_i \prod_{j \neq i} f_{j \mid \pa{j}}.$
It happens that if $f$ factors according to a tree rooted at some vertex, it factors according to that same tree rooted at any vertex, and we can therefore say \term{$f$ factors according to $T$} without ambiguity (see \cite{murphy2012machine}).
We call such densities \term{tree densities}.

A weighted graph $(G, W)$ is an undirected graph $G = (V, E)$ along with a weight function $W: E \to \R$.
The \term{weight} of a subgraph $(V, F)$ where $F \subset E$ is $\sum_{e \in F} W(e)$.
A subgraph of $G$ spans $G = (V, E)$ if all vertices in $V$ are connected.
A graph is a \term{forest} if it has no cycles.
If a forest $F$ spans a graph $G = (V, E)$ then $(V, F)$ is a tree.
A \term{maximum spanning tree} of $G$ with respect to $W$ is one whose weight is at least as large as that of all other trees which span $G$.

\section{Related Work}
\label{section:related}

There is an extensive literature using graphs to describe joint probability distributions and model functional relationships between multivariate data.
Standard texts on \term{probabilistic graphical models} include
\cite{pearl1988probabilistic},
\cite{lauritzen1996graphical}, 
\cite{koller2009probabilistic},
and \cite{murphy2012machine}.
In this setting, one uses the graph to describe the factoring properties of a distribution.
There are two variants of models.
The first uses directed acyclic graphs and the models are called Bayesian networks.
The second uses undirected graphs and the models are called Markov random fields.
For trees, these coincide. 
A distribution factors according to an (undirected) tree $T$ if and only if it factors according to every rooted (directed) tree $(T, r)$ for $r$ a vertex of $T$ \cite{murphy2012machine}.

In both cases, one associates the component random variables with nodes of the graph.
The graphical structure encodes the conditional independence relations of the variables.
It is a basic task of machine learning to learn a graph $G$ along with a distribution factoring according to $G$ from data.
This task, called \term{structure learning}, is often posed as a maximum likelihood problem.
In general, it is intractable \cite{koller2009probabilistic}.
The case when $G$ is a tree, however, is a notable exception.

The use of trees to describe the factorization of a joint distribution $p$ can be traced to the celebrated work of Chow and Liu \citet{chow1968approximating}.
The authors posed and solved a problem to approximate an arbitrary distribution $q$ by a distribution $p$ which factors according to a tree.
The key quantity is the mutual information between pairs of variables, and the solution method involves finding a maximal spanning tree of a graph whose edges are weighted by mutual informations.

The results of Chow and Liu also solved the \term{structure learning} problem in the case of trees.
The approximation criterion used is the Kullback-Leibler divergence of $p$ with respect to $q$.
In the case that $q$ is an empirical distribution of some dataset, the K-L divergence is the likelihood plus a constant.
This result has had widespread application \cite{friedman1997bayesian,friedman1998bayesian,meila2000learning}.
Later authors considered learning the tree structure specifically for discrimination \cite{tan2010learning} or using the Akaike information criterion \cite{edwards2010selecting}.
Recent work has compared these and other techniques for selecting the tree structure \citet{perez2016discriminant}.

A closely related approach is structural equation models \cite{pearl2009causality}.
A \term{structural equation model} associates each random variable with a node in a directed acyclic graph and models each variable as a function of its parents (if any) in the graph and noise.
These models can be interpreted as Bayesian networks.
It is common in the literature to model the noise as Gaussian and additive.
As with graphical models, one is interested in \textit{simultaneously} learning the graphical structure along with the functional relations.

Structural equation models are useful when one suspects the  functional relations to be causal.
For this reason, the task of estimating the structure and functions from observational data is called \term{causal inference}.
Causal inference is fundamental in many disciplines \cite{glass2013causal,ramsey2010six, statnikov2012new}.
The identifiability of these models, however, is nontrivial \cite{rothenhausler2018causal}.

In the context of control and decision problems, recent work has looked to extend the ideas of graphical models to stochastic processes \cite{bach2004learning, materassi2010topological,tan2010learning,quinn2015directed}.
Our work is closest in spirit to that of Materassi and Innocenti  \citet{materassi2010topological}.
They study a stochastic process variant of tree linear cascades and give theoretical results concerning identification.
Related to this work, Tan and Willsky study a stochastic process variant of the Chow-Liu algorithm \cite{tan2011sample}.
The key object in both of these papers is a particular maximum spanning tree.
\section{Tree linear cascades}
\label{section:cascades}

Since structural equation models are generally not identifiable, it is natural to look for classes of them which have properties leading to partial identifiability.
In this section we introduce a class of tree linear structural equation models with an interesting maximum spanning tree property.

The subtlety of our results is that we make weak assumptions on the distributions of the error variables.
We do not assume that the error variables are independent.
Nor do we assume that they have a Gaussian distribution.
Both of these assumptions are common in the literature \cite{pearl1988probabilistic,rothenhausler2018causal}.

The novelty here is that we can obtain interesting results with the weaker assumption that the error variables are uncorrelated.
To our knowledge, this definition of tree linear cascades and Theorem~\ref{theorem:probabilistic:mst} do not appear in prior literature.
The development is partially similar to a stochastic process variant of tree linear cascades studied in \cite{materassi2010topological}.

\subsection{Definition}
Throughout this section let $(T, r)$ be a rooted tree on $\upto{d}$.
 Define $\graphpath(i, j)$ to be the set of directed edges on the path from $i$ to $j$ if $i \in \anc(j)$, and the empty set otherwise.
Define
\begin{equation}
    \sparse(T, r) = \Set{A \in \R^{d \times d}}{A_{ij} = 0 \text{ if } j \neq \pa{i}}.
    \label{equation:cascades:sparseTr}
\end{equation}
Elements of $\sparse(T, r) \subset \R^{d \times d}$ have the same sparsity as the directed adjacency matrix of $(T, r)$.
Given a zero-mean random vector $e: \Omega \to \R^d$, with diagonal positive definite covariance, and a matrix $A \in \sparse(T, r)$, we say that $x$ is a \term{tree linear cascade} on $e$ with respect to $A$ if
\begin{equation}
  x = Ax + e.
  \label{equation:cascades:defining}
\end{equation}
Notice the uncorrelated assumption on $e$, which is weaker than the often-made independence assumption.

One can interpret Equation~\eqref{equation:cascades:defining} as a particular linear recursive equation \cite{wermuth1980linear}.
Alternatively, by adding a joint independence assumption on $e$, one can interpret Equation~\eqref{equation:cascades:defining} as a usual linear structural equation model \cite{rothenhausler2018causal} or as a directed graphical model \cite{murphy2012machine}.
The directed graph in both cases corresponds to a rooted tree.
It is common when using structural equation models like \ref{equation:cascades:defining} to make causal assumptions.
In this case, one can interpret Equation~\eqref{equation:cascades:defining} as a structural (also known as functional) causal model \cite{pearl2009causality}.

A special case of this class is the familiar tree linear Gaussian densities.
It is easy to show that if $e$ has a Gaussian distribution, then $x$ is Gaussian and factors according to a tree.
Conversely, if a Gaussian density $f$ factors according to tree $T$, its inverse covariance matrix has the sparsity pattern of $T$ \cite{lauritzen1996graphical}.
Using a sparse Cholesky factorization \cite{vandenberghe2015chordal}, one can show that there is an uncorrelated Gaussian random vector $e$, vertex $i$ of $T$, and $A \in \sparse(T, i)$ so that if $x = Ax + e$, then $x$ has density $f$.
In other words, one can obtain a tree structural equation model in which $e$ is Gaussian.

On the other hand, if $e$ does not have a Gaussian density, then $x$ is not Gaussian.
Consequently, the tree Gaussians are a strict subset of the densities representable by tree linear cascades.
In other words, we are considering a more general class of distributions than the tree linear Gaussians.
Herein lies the challenge and interest of our results.

\subsection{Tree, not root, of tree linear cascade is identifiable}
\label{section:cascades:root}

The key property of tree linear cascades is that the undirected tree structure is identifiable (Theorem~\ref{theorem:probabilistic:mst}).
The novelty and particular interest of this work is that we can show this property regardless of the distribution of each individual $e_i$.

Regardless of distributional assumptions, there is always an identifiability caveat when working with tree linear cascades.
For a tree linear cascade $x$ on $e$ with respect to $A \in \sparse(T, r)$,
one can show that there exists $e'$ and $A' \in \sparse(T, i)$ where $i \neq r$ is any other vertex of $T$ and $x$ is a tree linear cascade on $e'$ with respect to $A'$.
Therefore only the tree $T$ is identifiable, not the root.

In practice, this difficulty is alleviated by extra knowledge.
We are often modeling measurements from a system and we know the input's identity.
The input might be the load on a computer system and the other quantities might be various resource or application-specific metrics.
Such knowledge enables one to select the root and orient the edges.

\subsection{Identifying the tree}

The tree of a tree linear cascade can be identified as the maximum spanning tree of a graph using squared correlations as weights.
It is worth restating that this result holds regardless of distributional assumptions on $e$.

We require a progression of algebraic and graphical lemmas to show this result.
To streamline the presentation, these appear in the Appendix.
Although we have structured them so that the proof here appears short and straightforward, a glance at the Appendix indicates the work involved.
The algebraic properties of Equation~\eqref{equation:cascades:defining} and the uncorrelated components of $e$ give the following.
The hypothesis that $\E(x_i^2) = 1$ is not substantial to the result, and is made for convenience of the proofs.

\begin{theorem}
  Let $e: \Omega \to \R^d$ be a random vector with zero mean and diagonal positive definite covariance.
  Let $x: \Omega \to \R^d$ be a tree linear cascade on $e$ with respect to $A \in \sparse(T, r)$ and suppose $\E(x_i^2) = 1$.
  Then $T$ is the unique maximum spanning tree of the complete graph on $\upto{d}$ with edge $\set{i, j}$ weighted by $\E(x_ix_j)^2$.
  \label{theorem:probabilistic:mst}
\end{theorem}
\begin{proof}
  Define $C = \cov(x) = \cov((I-A)^{-1}e) = (I-A)^{-1}\cov(e)(I - A)^{-\tp}$. Then $\abs{C_{ij}} = \abs{\E(x_ix_j)}$ and $\cov(e)$ is diagonal positive definite.
  Use Lemma~\ref{lemma:algebraic:Cbound} to see that the conditions of Lemma~\ref{lemma:graphical:mst} are satisfied, and so conclude that $T$ is a minimum spanning tree of the complete graph on $\upto{d}$ where edge $\set{i, j}$ is weighted by $\abs{C_{ij}}$.
  Since squaring the correlation magnitudes $\abs{C_{ij}}$ is a monotonic transformation of the weights, a tree is a maximum spanning tree with edges weighted by $\abs{C_{ij}}$ if and only if it is the maximum spanning tree with edges weighted by $C_{ij}^2$ (using Lemma~\ref{lemma:monotone_mst}).
\end{proof}

\section{Simultaneous cascade regression}
\label{section:regression}

Structural equation and causal models are useful when one is working with multivariate data and believes that some of the variables are functionally related to others\cite{pearl2009causality}.
In some cases, extra knowledge enables one to specify the graphical structure of the model \cite{landolfi2021cloud}.
In practice, however, one frequently wants to learn the model structure.
Unfortunately, these models are not always identifiable\cite{rothenhausler2018causal}.

In practice, one does have data.
It is natural to simultaneously search for a directed graph and estimate the functions involved by regressing variables on their parents in the graph.
Beyond difficulty in estimating the functional relations, such problems are difficult because the number of directed acyclic graphs (even directed acyclic trees) is exponential in the number of variables.

It is pleasantly surprising, therefore, that if we restrict to linear functions and directed trees (as we do in Problem~\ref{problem:cascades:regression}), one can analytically solve the problem.
Although the statement of this problem and its solution remind one of the more specialized result of Chow and Liu for Gaussian densities, it is important to note here that we make no distributional assumptions on $x$.
To our knowledge, this  problem and its solution do not appear in the prior literature.

\begin{problem}[Simultaneous cascade regression]
  Suppose $x: \Omega \to \R^d$ is a random vector with $\E(x) = 0$ and $\E(x_i^2) = 1$.
  Find a rooted tree $(T, r)$ on $\upto{d}$ and $A \in \R^{d \times d}$ to
  \[
    \begin{aligned}
      \text{ minimize } &\quad \E\norm{Ax - x} \\
      \text{ subject to } &\quad A \in \sparse(T, r).
    \end{aligned}
  \]
  \label{problem:cascades:regression}
  We call a solution $T^\star$ an \term{optimal cascade tree} of $x$.
\end{problem}

Problem~\ref{problem:cascades:regression} finds a rooted tree which, if we estimated each component of $x$ using only its parent in the tree, gives the smallest expected sum of squared errors.
Notice that the sparsity constraint on $A$ ensures that the diagonal of $A$ is 0, and so $A = I$ is not a solution.

It is natural to ask that, if $x$ is a tree linear cascade with respect to $T$, any problem (like Problem~\ref{problem:cascades:regression}) which proposes to select a tree by which to model $x$ should recover $T$.
This statement is the content of Corollary~\ref{corollary:regression:correctness} in Section~\ref{section:regression:identification}.
In the case that $x$ is not a tree linear cascade, we can still use our solution to Problem~\ref{problem:cascades:regression}, and interpret Problem~\ref{problem:cascades:regression} as a variational principle for selecting a tree when modeling a random vector.

In practice, we are interested in the empirical risk minimization form of Problem~\ref{problem:cascades:regression} (see Section~\ref{section:empirical}).
In that setting, we do not expect data generated by a tree linear cascade.
Instead, we view Problem~\ref{problem:cascades:regression} as a fitting method.

\subsection{Solution of simultaneous cascade regression}

We solve Problem~\ref{problem:cascades:regression} in two pieces.
First, we find the optimal cascade coefficients for a rooted tree $(T, r)$ (Lemma~\ref{lemma:regression:optimalA}).
Then we find the optimal tree (Theorem~\ref{theorem:regression:optimalT}).

\begin{lemma}
  Suppose $x: \Omega \to \R^d$ is a random vector with $\E(x) = 0$ and $\E(x_i^2) = 1$.
  Let $(T, r)$ be a rooted tree on $\upto{d}$.
  Define $A^\star \in \sparse(T, r)$ by
  \[
    A^\star_{ij} = \begin{cases}
      \E(x_ix_j) & j = \pa{i} \\
      0 & \text{otherwise.}
    \end{cases}
  \]
  Then $A^\star$ minimizes $\norm{Ax - x}$ among all $A \in \sparse(T, r)$.
  \label{lemma:regression:optimalA}
\end{lemma}
\begin{proof}
  Let $A \in \sparse(T, r)$.
  Express
  \[
    \E \norm{Ax - x}^2 = \E(x_{r}^2) + \sum_{i \neq r}\E\parens{A_{i\pa{i}}x_{\pa{i}} - x_{i}}^2.
  \]
  The first term does not depend on $A$ and the second sum separates across $i$.
  For $i \neq r$ we find $A_{i\pa{i}}$ to minimize $\E\parens{A_{i\pa{i}}x_{\pa{i}} - x_i}^2$.
  A solution is $A^{\star}_{i\pa{i}} = \E(x_ix_{\pa{i}})\inv{\E(x_{\pa{i}}^2)} = \E(x_ix_{\pa{i}})$.
\end{proof}

\begin{theorem}
 	Suppose $x: \Omega \to \R^d$ is a random vector with $\E(x) = 0$ and $\E(x_i^2) = 1$.
	A tree $T$ on $\upto{d}$ is an optimal cascade tree of $x$ if and only if it is a maximum spanning tree of the complete graph on $\upto{d}$ with edge $\set{i, j}$ weighted by $\E(x_ix_j)^2$.
  \label{theorem:regression:optimalT}
\end{theorem}
\begin{proof}
  Let $r \in \upto{d}$.
  By Lemma~\ref{lemma:regression:optimalA}, there exists $A^{\star} \in \sparse(T, r)$ which minimizes $\norm{Ax - x}$ among $A \in \sparse(T, r)$.
  We have
  \[
    \begin{aligned}
      \E\norm{A^{\star}x - x}^2
      &= \E(x_{r}^2) + \sum_{i \neq r}\E\parens{\E\parens{x_ix_{\pa{i}}}x_{\pa{i}} - x_{i}}^2 \\
      &= \E(x_{r}^2) + \sum_{i \neq r} \E(x_i^2) - \E(x_ix_j)^2 \\
      &= d - \sum_{i \neq r} \E(x_ix_j)^2.
    \end{aligned}
  \]
Here $d$ is a constant, and the second term is a sum over the edges of $T$ (and does not depend on the root).
To minimize the sum, we choose $T$ to be a maximum spanning tree with weights $\E(x_ix_j)^2$.
\end{proof}

Two aspects of this derivation stand out.
First, the proof of this theorem reminds one of the classical Chow-Liu result for Gaussians \cite{chow1968approximating, tan2010learninggaussian}.
We make a connection for the special case of a Gaussian tree linear cascade to the Chow-Liu algorithm in Section~\ref{section:chow}.
For now, we reiterate that so far we have made no Gaussian distributional assumptions.
Second, the solution of Problem~\ref{problem:cascades:regression} does not depend on the choice of root.
We address this in the next section.

\subsection{Problem~\ref{problem:cascades:regression} recovers tree of tree linear cascade}
\label{section:regression:identification}

It is natural to expect that if the vector $x$ in Problem~\ref{problem:cascades:regression} is a tree linear cascade on $e$ with respect to $A \in \sparse(T, r)$, then $(T, r)$ should be a solution.
As discussed in Section~\ref{section:cascades:root}, however, we can not hope to identify the root.
Equipped with Theorem~\ref{theorem:probabilistic:mst}, one can show the following.

\begin{corollary}
  Let $(T, r)$ be a rooted tree on $\upto{d}$ and $e: \Omega \to \R^d$ a random vector with zero-mean and diagonal positive definite covariance.
  Suppose $x: \Omega \to \R^d$ is a tree linear cascade on $e$ with respect to $A \in \sparse(T, r)$ and that $\E(x_i^2) = 1$.
  Then $T$ is the unique optimal cascade tree of $x$.
  \label{theorem:simultaneous:correctness}
  \label{corollary:regression:correctness}
\end{corollary}
\begin{proof}
Express $\E(x) = (I - A)^{-1}\E(e) = 0$.
Use Theorem~\ref{theorem:regression:optimalT} to conclude that an optimal cascade tree is a maximum spanning tree of the complete graph on $\upto{d}$ weighted by $\E(x_ix_j)^2$.
Use Theorem~\ref{theorem:probabilistic:mst} to conclude that $T$ is the only maximum spanning tree.
\end{proof}

\section{Connections to Gaussian Chow-Liu}
\label{section:chow}

Sections~\ref{section:cascades} and \ref{section:regression} present the primary theoretical contributions of this paper.
These results remind one of the classical results  for approximating a density by one which factors according to a tree \citet{chow1968approximating}.
In this section, we briefly state these well-known results and then show how they connect to our results for the special case of Gaussian densities.

\subsection{Review of well-known Chow-Liu results}

We include formal statements to make precise our treatment in Section~\ref{section:chow:regression}. 
The following problem and its solution are well-known \cite{tan2010learninggaussian}.
\begin{problem}[Tree density approximation]
  Given a density $g: \R^d \to \R$, find a tree $T$ on $\upto{d}$ and a density $f: \R^d \to \R$ to
  \[
    \begin{aligned}
      \text{minimize} \quad& d_{kl}(g, f) \\
      \text{subject to} \quad& f \text{ factors according to } T.
    \end{aligned}
  \]
  We call a solution $T^\star$ an \term{optimal approximator tree} of $g$.
	\label{problem:chow:approximation}
\end{problem}





\begin{lemma}
  Let $g$ be a density on $\R^d$.
  A tree $T$ on $\upto{d}$ is an optimal approximator tree of $g$ if and only if it is a maximum spanning tree of the complete undirected graph on $\upto{d}$ with edge $\set{i, j}$ weighted by $d_{kl}(g_{ij}, g_ig_j)$.
\label{lemma:chow:optimalT}
\end{lemma}


%

\subsection{Novel connection to cascade regression}
\label{section:chow:regression}

Chow-Liu for densities is computationally feasible if one makes a Gaussian assumption.
Unsurprisingly, therefore, this assumption is extremely common in the literature.
In this case, Problem~\ref{problem:chow:approximation} reduces to finding the maximum spanning tree of a graph weighted by $-\nicefrac{1}{2}\log(1 - \rho_{ij}^2)$. This quantity is the mutual information between two components of a Gaussian.

In this Gaussian assumption lies the connection to tree linear cascades and to cascade regression (Problem~\ref{problem:cascades:regression}).
With the results of Section~\ref{section:regression} and brief reflection on the properties of maximum spanning trees, one can show that the spanning tree found by Chow and Liu is the same as that in Section~\ref{section:regression}.


\begin{corollary}
  Let $x: \Omega \to \R^d$ be a random vector with $\E(x_i) = 0$ and $\E(x_i^2) = 1$ for $i \in \upto{d}$.
  Suppose $x$ has a Gaussian density $f: \R^d \to \R$.
  A tree $T$ is an optimal cascade tree of $x$ if and only if it is an optimal approximator tree of $f$.
  \label{corollary:chow:correspondence}
\end{corollary}
\begin{proof}
By Lemma~\ref{lemma:chow:optimalT}, an optimal approximator tree of $f$ is a maximal spanning tree of the complete graph on $\upto{d}$ with edge $\set{i, j}$ weighted by $d_{kl}(f_{ij},f_{i}f_j) = -(\nicefrac{1}{2})\log(1 - \E(x_ix_j)^2)$.
By Theorem~\ref{theorem:regression:optimalT}, an optimal cascade tree of $x$ is a maximal spanning tree of the complete graph on $\upto{d}$ weighted by $\E(x_ix_j)^2$.
The sets of maximal spanning trees coincide because the first quantity $d_{kl}(f_{ij}, f_if_j)$ is a monotonic transformation of the squared correlations $\E(x_ix_j)^2$.
\end{proof}

In Section~\ref{section:regression:identification}, Corollary~\ref{corollary:regression:correctness} says that cascade regression identifies the tree for all tree linear cascades.
Of course, the special case of Gaussian tree linear cascades is included.
Corollary~\ref{corollary:chow:correspondence}, therefore, can be roughly interpreted as strengthening the justification for the Gaussian version of the classical Chow-Liu result.
 One has theoretical recourse to Corollary~\ref{corollary:chow:correspondence}, which says that one may interpret the approach as making a linear assumption while avoiding a Gaussian assumption.

This is intuitively reminiscent of, but distinct from, a result in classical estimation.
An \term{estimator} for a random $m$-vector $y$ from a random $n$-vector $x$ is a function $\phi: \R^n \to \R^m$.
If we select $\phi$ to minimize $\E\norm{\phi(x) - y}$, the solution $\phi^\star$ is the conditional expectation of $y$ given $x$.
On one hand, if we assume that the joint density of $x$ and $y$ is Gaussian, the conditional expectation of $y$ given $x$ is  $\phi(x) = Ax + b$ for some $A$ and $b$.
On the other hand, if we constrain $\phi(x) = Cx + d$ for some $C$ and $d$, 
 $C$ and $d$ are expressible in terms of the covariance and means of $x$ and $y$.

It is well-known that if one takes the latter approach, and $x$ and $y$ are jointly Gaussian, then $\opt{C} = A$ and $\opt{d} = b$.
Roughly speaking, Gaussian Chow-Liu and cascade regression are linked in a similar war.
In one case we impose a linear constraint.
In the other we make a distributional assumption that reduces to a linear constraint.
Instead of the conditional expectation, the central quantity is the pairwise mutual informations between pairs of random variables.

\section{Empirical Cascade Regression}
\label{section:empirical}

\begin{figure}
	\includegraphics[width=0.48\textwidth]{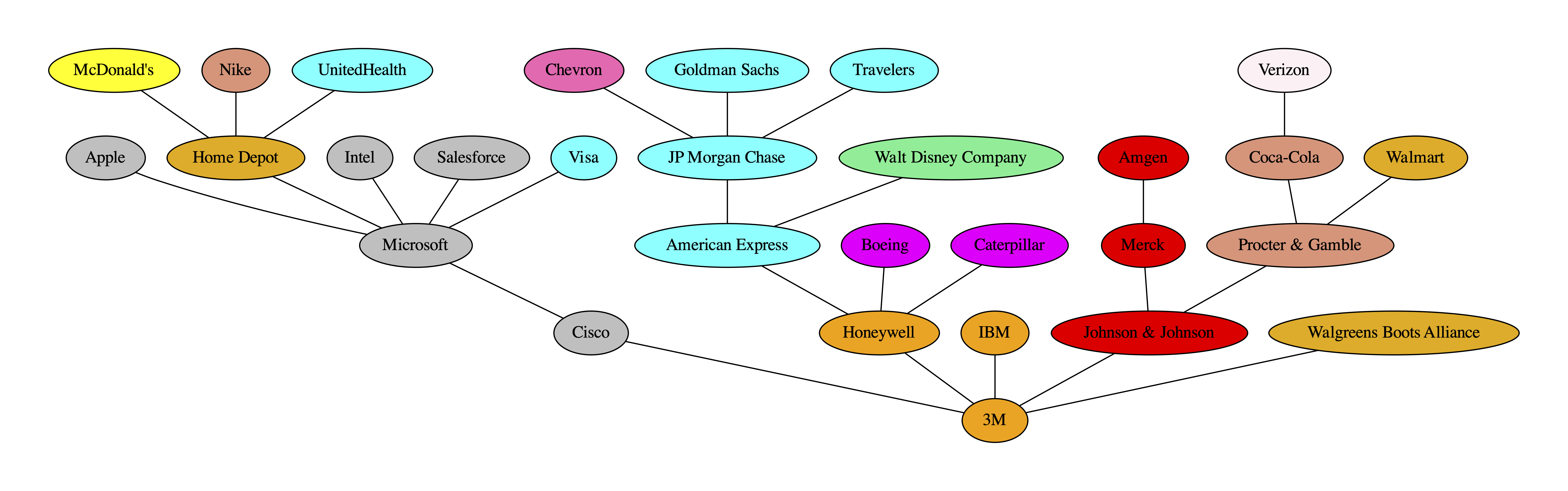}
	\centering
	\caption{
	The optimal cascade tree for a dataset of stock prices.
        }
	\label{figure:empirical:stocks}
        \vspace*{-6mm}
\end{figure}

In practice, we have data and want to fit a structural equation model.
Cascade regression
(Problem~\ref{problem:cascades:regression}) has a natural instantiation as empirical risk minimization.
The quantities involved are easily computable, and we give a straightforward example on stock prices.

\subsection{Cascade regression as empirical risk minimization}

Let $x^1, \dots, x^n$ be a dataset in $\R^d$.
For \term{empirical cascade regression}, one finds a rooted tree $(T, r)$ on $\upto{d}$ and coefficients $A$ to minimize $\sum_{k = 1}^{n} \norm{Ax^k - x^k}$ subject to $A \in \sparse(T, r)$.

The key quantity for finding the optimal cascade tree is the squared empirical correlations.
We use any standard maximum spanning tree algorithm (e.g., \cite{kruskal1956shortest}).
As discussed in Section~\ref{section:chow:regression}, the tree found is the same as if one used the Chow-Liu algorithm with a Gaussian assumption.


\subsection{Simple stock price movement example}
\label{section:data_based:stocks}

To illustrate our theoretical results and show that it is straightforward to compute the quantities involved, we include a small example on stock prices.
Our implementation (using Julia \cite{bezanson2017julia}) is short and requires trivial computation.
We collect daily price changes (for the past decade) of the thirty stocks of the Dow Jones Industrial Average from the public \href{https://www.wsj.com/market-data}{market data page} of the Wall Street Journal.
For simplicity, we ignore time dependence.
We make the dataset and code available.

We visualize the tree structure recovered in Figure~\ref{figure:empirical:stocks}.
Stocks are colored by their industry according to the Dow Jones \& Company's classification.
Roughly, stocks in similar industries are connected in the graph.
A well-known and natural clustering method for $k \geq 2$ clusters is to delete the $k-1$ lightest edges from this tree.
Of course, this technique is not sufficient to obtain state of the art predictions for the stock market.
The applicability of tree linear cascades for a particular data set must be evaluated on a case by case basis.

\section{Conclusion}
\label{section:conclusion}

We define tree linear cascades and show that their tree structure  is identifiable as the solution of a maximum spanning tree problem (Theorem~\ref{theorem:probabilistic:mst}).
In a parallel vein, we discuss a natural regression problem for fitting structural equation models, and show that a constrained form of this problem (Problem~\ref{problem:cascades:regression}) has an analytical solution (Theorem~\ref{theorem:regression:optimalT}).
We connect these results to the more specialized results of Chow and Liu for Gaussian densities.
We conclude with a simple data-based example demonstrating that the quantities involved are easy to compute.

\paragraph{Limitations and future work}
The upside of tree linear cascades is that they are amenable to mathematical analysis.
The downside, as usual, is that they can not model all distributions.
The restriction on the distributions represented has two components.
One is the linearity and one is the assumption of tree structure.
We plan to alleviate the former of these limitations by extending these results to the block case.
One can embed measurements using features and so approximate nonlinear relationships between the variables with linear functions of nonlinear features.
Another limitation is that the root of tree linear cascades is not identifiable.
In practice, however, one often has extra information about which variable should be the input.
Future theoretical work may analyze sample complexity.
Future applied work may look at particular application domains (e.g., genetics and computer system modeling).


\section*{APPENDIX}

\subsection{Graphical results relating to maximum spanning trees}

Let $(G, W)$ be a weighted graph with vertex set $\upto{d}$.
As usual, a function $h: \R \to \R$ is monotone increasing if $h(x) < h(y)$ whenever $x < y$ for all $x, y \in \R$.
We use the following in Theorem~\ref{theorem:probabilistic:mst} and Corollary~\ref{corollary:chow:correspondence}.
\begin{lemma}
\label{lemma:monotone_mst}
Suppose $h: \R \to \R$ is monotone increasing.
Define $H_{e} = h(W_e)$ for each edge $e$.
A tree $T$ is a maximum spanning tree of $(G, W)$ if and only if it is a maximum spanning tree of $(G, H)$.
\end{lemma}
\begin{proof}
See Section 1.1.12 of \cite{marevs2008saga}.	
\end{proof}


The following construction handles nonunique weights.
Let $\beta_1, \beta_2, \dots$ be a sequence of edges of $G$. 
$\beta$ is consistent with $W$ if $W_{\beta_{a}} \geq W_{\beta_{b}}$ whenever $a < b$ for integers $a$ and $b$.
The \term{Kruskal edge} of a forest $F$ with respect to $\beta$ is the first edge in $\beta$ which is not in $F$ and whose addition to $F$ creates no cycles.
The \term{Kruskal sequence} corresponding to $\beta$ is the sequence of forests $F_0^{\beta} \subset F_1^{\beta} \subset \dots \subset F^{\beta}_{d-1}$ where $F_0^{\beta} = \emptyset$ and each forest differs with the subsequent one only by its Kruskal edge.
The following is a straightforward variant of \cite{kruskal1956shortest}.


\begin{lemma}
  Let $T$ be a maximum spanning tree of $(G, W)$.
  There exists an ordering $\beta$ consistent with $W$ whose Kruskal sequence $F^{\beta}_{0}, F^{\beta}_{1}, \dots, F^{\beta}_{d-1}$ satisfies $F^{\beta}_{d-1} = T$.
  \label{lemma:graphical:orders}
\end{lemma}

We use Lemma~\ref{lemma:graphical:orders} to get the uniqueness piece of Lemma~\ref{lemma:graphical:mst}. 
The hypothesis of Lemma~\ref{lemma:graphical:mst} corresponds to the conclusion of Lemma~\ref{lemma:algebraic:Cbound}.

\begin{lemma}
  Let $T$ be a tree on $\upto{d}$.
  Suppose that for every two distinct nonadjacent vertices $i$ and $j$ of $T$, and $k$ the first vertex on the path from $i$ to $j$,
  \begin{equation}
    W_{ij} < \min\set{W_{ik}, W_{kj}}.
    \label{equation:cascades:graphicalcondition}
  \end{equation}
  Then $T$ is the unique maximum spanning tree of $G$.
  \label{lemma:graphical:mst}
\end{lemma}
\begin{proof}
  Let $\tilde{T}$ be a maximum spanning tree of $G$.
  We will show that $\tilde{T} = T$.
  By Lemma~\ref{lemma:graphical:orders}, there exists an order $\beta$ of edges of $G$ so that the Kruskal sequence $F_0^{\beta} \subset F_{1}^{\beta} \subset \dots \subset F_{d-1}^{\beta}$ satisfies $F_{d-1}^{\beta} = \tilde{T}$.
  Define $\tilde{t}_1, \dots, \tilde{t}_{d-1}$ to be the corresponding sequence of Kruskal edges.
  Suppose, toward contradiction, that there is an edge of $\tilde{T}$ not in $T$.
  Let $\tilde{t}_s = \set{i, j}$ be the first such edge and  let $k$ be the first vertex on the path from $i$ to $j$ in $T$.
  By hypothesis, $W_{ik} > W_{ij}$ and $W_{kj} > W_{ij}$.
  So by construction these edges are either in $F^{\beta}_{s-1}$ or were skipped because the vertices involved were already connected.
  Either way,  $i$ is connected to $k$ and $k$ is connected to $j$ in in $F^\beta_{s-1}$.
  But then there is a cycle in $\tilde{T}$. $i$ to $k$, $k$ to $j$ and $j$ to $i$.
  So $\tilde{T}$ is not a tree, a contradiction.
\end{proof}

We can interpret the condition in Equation~\ref{equation:cascades:graphicalcondition} by way of a variant of the algorithm in \cite{kruskal1956shortest}.
This variant constructs a Kruskal forest by considering the edges in decreasing order of weight.
The condition says, roughly, that before the algorithm considers connecting nonadjacent vertices $i$ and $j$, it will already have considered connecting $i$ with $k$, the first vertex on the path to $j$, and $k$ with $j$.

\subsection{Algebraic structure of tree linear cascades}
\label{section:appendix:algebraic}

Since Equation~\ref{equation:cascades:defining} implies $(I - A)x = e$, we analyze	 $(I - A)^{-1}$.
The following two lemmas are special cases of standard results \cite{meyer2000matrix}.

\begin{lemma}
  Suppose $A \in \sparse(T, r)$. Let $p \geq 1$.
  Then $A^p_{ij}$ is $\prod_{(s, t) \in \graphpath(j,i)} A_{ts}$ if $\abs{\graphpath(j,i)} = p$ and $0$ otherwise.
  \label{lemma:algebraic:powers}
\end{lemma}

\begin{lemma}
  Suppose $A \in \sparse(T, r)$. Then $(I - A)^{-1}$ exists and
  $((I - A)^{-1})_{ij}$ is $A_{ij}^{\abs{\graphpath(j, i)}}$ if there is a directed path from $j$ to $i$, 1 if $i = j$, and 0 otherwise.
  \label{lemma:algebraic:inverse}
\end{lemma}

\begin{proof}
  From Lemma~\ref{lemma:algebraic:powers} it follows that $A$ is nilpotent, and so $(I - A)^{-1}$ exists and is given by its Neumann series $\sum_{k = 0}^{\infty} A^k$ (see p. 126 of \cite{meyer2000matrix}). 
\end{proof}

An immediate result of Lemma~\ref{lemma:algebraic:inverse} is a useful path-factoring property of elements of $(I - A)^{-1}$.

\begin{lemma}
  Suppose $A \in \sparse(T, r)$.
  Define $B = (I - A)^{-1}$.
  If there is a directed path from $j$ to $i$ and $k$ is a vertex on it, then $B_{ij} = B_{ik}B_{kj}$.
  \label{lemma:algebraic:factors}
\end{lemma}

The next four lemmas are motivated by $\cov(x) = (I - A)^{-1}\cov(e)(I - A)^{-\tp}$.
The first two give the covariance between two components of $x$.
First when one component is an ancestor of the other (Lemma~\ref{lemma:algebraic:Cijancestor}) and second when neither is an ancestor of the other (Lemma~\ref{lemma:algebraic:Cijseparate}).
We then bound the magnitude of the elements of $A$ by 1 (Lemma~\ref{lemma:algebraic:Abound}).
Finally, we use these to show a result (Lemma~\ref{lemma:algebraic:Cbound}) which says, roughly, that a component of a tree linear cascade is most correlated with its neighbors in the tree.
The unit diagonal hypothesis matches the unit variance assumption in Theorem~\ref{theorem:probabilistic:mst}.

\begin{lemma}
  Let $A \in \sparse(T, r)$ and $D$ diagonal positive definite.
  Define $B = (I - A)^{-1}$ and $C = BDB^\tp$.
  Suppose $\diag(C) = 1$. If $i \in \anc(j)$, then $C_{ij} = C_{ji} = B_{ji}$.
  \label{lemma:algebraic:Cijancestor}
\end{lemma}
\begin{proof}
  Define $F = BD^{1/2}$.
  The sparsity pattern of $F$ is the same as that of $B$.
  Define $\anc^{+}(i) = \anc(i) \union \set{i}$.
  Use Lemma~\ref{lemma:algebraic:inverse}
  to express
  \[
  \begin{aligned}    C_{ij}
    = \sum_{k = 1}^{d} F_{ik}F_{jk} 
    &= \sum_{k \in \anc^+(i)} F_{ik}F_{jk} \\
    &= \sum_{k \in \anc^+(i)} F_{ik}\parens{B_{jk}D_{kk}^{1/2}}.
  \end{aligned}
  \]
  Since there is a directed path from every vertex in $\anc^+(i)$ to $j$, and $i$ is on it, use Lemma~\ref{lemma:algebraic:factors} to express
  \[
  	\begin{aligned}
    \sum_{k \in \anc^+(i)} F_{ik} \parens{B_{jk}D_{kk}^{1/2}}
    &= \sum_{k \in \anc^+(i)} F_{ik} \parens{B_{ji}B_{ik}D_{kk}^{1/2}}\\
    &= B_{ji} \sum_{k \in \anc^+(i)} F_{ik}^2.
    \end{aligned}
  \]
  Since $\sum_{k \in \anc^+(i)} F_{ik}^2 = C_{ii} = 1$ and $C$ is symmetric,  $C_{ij} = C_{ji} = B_{ji}$.
\end{proof}
\begin{lemma}
  Let $A \in \sparse(T, r)$ and $D$ diagonal positive definite.
  Define $B = (I - A)^{-1}$ and $C = BDB^{\tp}$.
  Suppose $\diag(C) = 1$. If $i,j$ are two distinct vertices and $i \not\in \anc(j)$ and $j \not\in \anc(i)$, then $C_{ij} = B_{im}B_{jm}$ where $m$ is the last coinciding vertex on the directed paths $r$ to $i$ and $r$ to $j$.
  \label{lemma:algebraic:Cijseparate}
\end{lemma}
\begin{proof}
  Define $\canc(i, j) = \anc(i) \cap \anc(j)$:
  Since neither $i$ nor $j$ is the root, $r \in \canc(i, j)$.
  Since $m$ is on the path from $r$ to $i$ and $r$ to $j$, $m \in \canc(i, j)$.
  In fact, $\canc(i, j) = \set{m} \union \anc(m)$.
  Define $F = BD^{1/2}$.
  The sparsity pattern of $F$ is the same as that of $B$.  Use Lemma~\ref{lemma:algebraic:inverse}
  and Lemma~\ref{lemma:algebraic:factors}
  to express
  \[
  \begin{aligned}
      C_{ij}
    = \sum_{k = 1}^{d} F_{ik}F_{jk}
    &= \sum_{k \in \canc(i, j)} F_{ik}F_{jk} \\
    &=\sum_{k \in \canc(i, j)} B_{im}F_{mk}B_{jm}F_{mk}.
    \end{aligned}
  \]
  From $\sum_{k \in \canc(i, j)} F_{mk}^2 = C_{mm} = 1$, conclude that $C_{ij} = B_{im}B_{jm}$.
\end{proof}

\begin{lemma}
  Let $A \in \sparse(T, r)$ and $D$ diagonal positive definite.
  Define $B = (I - A)^{-1}$ and $C = BDB^\tp$.
  Suppose $\diag(C) = 1$. Then $\abs{A_{ij}} < 1$.
  \label{lemma:algebraic:Abound}
\end{lemma}

\begin{proof}
  Define $F = BD^{1/2}$.
  The sparsity pattern of $F$ is the same as that of $B$.
  For any $i \in \upto{d}$, use this sparsity (see Lemma~\ref{lemma:algebraic:inverse}) to express
  \begin{equation}
    C_{ii} = \sum_{s = 1}^{d} F_{is}^2 = \sum_{s \in \anc(i)} F_{is}^2 + D_{ii}.
    \label{equation:algebraic:Cii}
  \end{equation}
  Use Lemma~\ref{lemma:algebraic:factors} and $F_{is} = B_{is}D^{1/2}_{ss}$ to express $\sum_{s \in \anc(i)} F_{is}^2$ as
  \[
     \sum_{s \in \anc(s)} \parens{A_{i\pa{i}}B_{\pa{i}s}D_{ss}}^2 = A_{i\pa{i}}^2\sum_{s \in \anc(i)} F_{is}^2.
  \]
  Since $\anc(i) = \anc(\pa{i}) \union \set{\pa{i}}$, $\sum_{s \in \anc(i)} F_{is}^2 = C_{\pa{i}\pa{i}} = 1$.
  Use these facts and Equation~\ref{equation:algebraic:Cii} to conclude
  $C_{ii} = A_{i\pa{i}}^2 + D_{ii}$.
  Since $C_{ii} = 1$, and $D_{ii} > 0$, the foregoing expression implies
  $\abs{A_{i\pa{i}}} < 1$.
\end{proof}

\begin{lemma}
  Let $A \in \sparse(T, r)$ and $D$ diagonal positive definite.
  Define $B = (I - A)^{-1}$ and $C = BDB^\tp$.
  Suppose $\diag(C) = 1$.
  If $i$ and $j$ are two distinct nonadjacent vertices in $T$, then the first vertex $k$ on the undirected path from $i$ to $j$ satisfies $\abs{C_{ij}} < \min\set{\abs{C_{ik}},\abs{C_{kj}}}$.
  \label{lemma:algebraic:Cbound}
\end{lemma}

\begin{proof}
  By cases:

  First, suppose $i \in \anc(j)$.
      In this case, $k$ is a child of $i$ in $(T, r)$.
      Deduce
      \begin{equation}
        \begin{aligned}
          \abs{C_{ij}}
          \overset{(a)}{=} \abs{B_{ji}}
          \overset{(b)}{=} \abs{B_{jk}B_{ki}}
          \overset{(c)}{<} \abs{B_{ki}}
          \overset{(d)}{=} \abs{C_{ik}}.
        \end{aligned}
        \label{equation:algebraic:Cijancestor}
      \end{equation}
      where (a) uses Lemma~\ref{lemma:algebraic:Cijancestor}, (b) uses Lemma~\ref{lemma:algebraic:factors}, (c) uses $\abs{B_{jk}} < 1$ from
      Lemma~\ref{lemma:algebraic:inverse} and Lemma~\ref{lemma:algebraic:Abound}, and (d) uses Lemma~\ref{lemma:algebraic:Cijancestor}.
      Similarly, $\abs{B_{ki}} < 1$ and so $\abs{C_{ij}} < \abs{B_{jk}} = \abs{C_{kj}}$.
      
Second, suppose $j \in \anc(i)$.
      In this case, $k$ is the parent of $i$ in $(T, r)$.
      Use the symmetry of $C$ and the previous case to conclude
      \[
        \abs{C_{ij}} = \abs{C_{ji}} < \min\set{\abs{C_{jk}}, \abs{C_{ki}}}
        = \min\set{\abs{C_{kj}}, \abs{C_{ik}}}.
      \]

Finally, suppose $i \not\in\anc(j)$ and $j \not\in\anc(i)$.
      In this case, $k$ is the parent of $i$ in $(T, r)$.
      Let $m$ be the last coinciding vertex on the paths from $r$ to $i$ and $r$ to $j$.
      Express
      \[
        \abs{C_{ij}} \overset{(a)}{=} \abs{B_{im}B_{jm}} \overset{(b)}{=} \abs{B_{ik}B_{km}B_{jm}} \overset{(c)}{<} \abs{B_{ik}} \overset{(d)}{=} \abs{C_{ik}}.
      \]
      where (a) uses Lemma~\ref{lemma:algebraic:Cijseparate}, (b) uses Lemma~\ref{lemma:algebraic:factors},  (c) uses $\abs{B_{km}B_{jm}} < 1$ from Lemma~\ref{lemma:algebraic:inverse} and Lemma~\ref{lemma:algebraic:Abound}, and (d) uses Lemma~\ref{lemma:algebraic:Cijancestor}.
      To show $\abs{C_{ij}} < \abs{C_{kj}}$, first suppose $k = m$. Then
      \[
        \abs{C_{ij}} \overset{(a)}{=} \abs{B_{ik}B_{jk}} \overset{(b)}{<} \abs{B_{jk}} \overset{(c)}{=} \abs{C_{kj}}.
      \]
      where
      (a) uses Lemma~\ref{lemma:algebraic:Cijseparate},
      (b) uses $\abs{B_{ik}} < 1$ from Lemma~\ref{lemma:algebraic:Abound},
      (c) uses Lemma~\ref{lemma:algebraic:Cijancestor}.
      Next suppose $k \neq m$: then $k \not\in \anc(j)$.
      Also, $j \not\in \anc(i)$ implies $j \not\in\anc(k)$.
      Express
      \[
        \abs{C_{ij}} \overset{(a)}{=} \abs{B_{ik}B_{km}B_{jm}} \overset{(b)}{<} \abs{B_{km}B_{jm}} \overset{(c)}{=} \abs{C_{kj}}.
      \]
      where (a) uses Lemma~\ref{lemma:algebraic:Cijseparate} and Lemma~\ref{lemma:algebraic:factors}, (b) uses $\abs{B_{ik}} < 1$ from Lemma~\ref{lemma:algebraic:Abound} and (c) uses Lemma~\ref{lemma:algebraic:Cijseparate}.
\end{proof}

\section*{ACKNOWLEDGMENT}
We thank our reviewers for helping us to improve the paper.
N. C. Landolfi is supported by a National Defense Science and Engineering Graduate Fellowship and a Stanford Graduate Fellowship.
S. Lall was partially supported by the National Science Foundation under grant 1544199.

\bibliographystyle{IEEETranS}
\bibliography{refs}

\end{document}